\documentclass[10pt,reqno,notitlepage]{article}

\usepackage{amssymb}
\usepackage{amsthm}
\usepackage{graphicx}
\usepackage{amsmath}
\usepackage{fancyhdr}
\usepackage{datetime}
\usepackage[toc]{appendix}
\usepackage[usenames,dvipsnames]{color}
\usepackage{tikz}
\usepackage{mathrsfs}
\usepackage[normalem]{ulem}
\usepackage{graphicx, xcolor}

\usepackage{setspace}
\usepackage{graphicx, xcolor}												
\usepackage[mathscr]{eucal}												
\usepackage{subcaption}
\usepackage{color}
\usepackage{epstopdf}

\usetikzlibrary{arrows,decorations.pathreplacing}
\usepackage[top=1in, bottom=1in, left=1in, right=1in]{geometry}
\usepackage{mathtools}									%
\mathtoolsset{showonlyrefs=true}							%

\newtheorem{theorem}{Theorem}
\newtheorem{lemma}[theorem]{Lemma}								%

\newtheorem{assumption}[theorem]{Assumption}	
\newtheorem{remark}{Remark}
\newtheorem{algorithm}[remark]{Algorithm}
\numberwithin{equation}{section}	
\numberwithin{theorem}{section}
\renewcommand{\(}{\left(}				
\renewcommand{\)}{\right)}
\renewcommand{\[}{\left[}
\renewcommand{\]}{\right]}

\def\R{\mathbb{R}}

\def\d{\partial}		

\def\ind{\mathbb{I}}

\def \abs#1{\left| #1 \right| }

\newcommand{\diag}{\operatorname{diag}}

\def\T{\top}

\DeclareMathOperator*{\argmin}{arg\,min}

\newtheorem{example}[theorem]{Example}

\long\def\symbolfootnote[#1]#2{\begingroup\def\thefootnote{\fnsymbol{footnote}}\footnote[#1]{#2}\endgroup}

\begin{document}

\title{Optimization of Fire Sales and Borrowing in Systemic Risk\footnote{We thank two anonymous referees for helpful comments.}}

\author{
Maxim Bichuch
\thanks{
Department of Applied Mathematics and Statistics,
Johns Hopkins Unversity
3400 North Charles Street, 
Baltimore, MD 21218. 
{\tt mbichuch@jhu.edu}. Research is partially supported by the Acheson J. Duncan Fund for the Advancement of Research in Statistics.}
\and  Zachary Feinstein
\thanks{
Department of Electrical and Systems Engineering, 
Washington University in St. Louis, St. Louis,
MO 63130, USA,
{\tt  zfeinstein@wustl.edu}. }
}
\date{\today}
\maketitle

\begin{abstract}
This paper provides a framework for modeling financial contagion in a network subject to fire sales and price impacts, but allowing for firms to borrow to cover their shortfall as well. We consider both uncollateralized and collateralized loans. The main results of this work are providing sufficient conditions for existence and uniqueness of the clearing solutions (i.e., payments, liquidations, and borrowing); in such a setting any clearing solution is the Nash equilibrium of an aggregation game.
\end{abstract}

{\bf AMS subject classification} 91G99, 90B10, 91A06.\\
\indent {\bf JEL subject classification} G32.\\
\indent {\bf Keywords} Systemic Risk, Networks, Fire Sales, Borrowing, Financial Contagion.\\

\section{Introduction}

Traditional financial risk considers each financial firm as separate and individual entities that do not interact or exacerbate each others downside events. Systemic risk, in contrast, considers the risk of the distress of a single bank or multiple banks spreading throughout the financial system, up to and including threatening the health of the entire system, due to characteristics of the interactions between firms. This spread of defaults is also called financial contagion. These contagious events can occur through local connections (e.g., contractual obligations) or global connections (e.g., impacts to asset prices). Such a systemic event occurred during the 2007-2009 financial crisis in which the entire financial system was threatened with failure. Due to the threat, this event led to government intervention requiring a significant bailout and directly precipitating a global recession. It is for these reasons that the modeling of systemic events is of paramount importance. This study will advance the modeling of such events by incorporating notions of borrowing and fire sales.

This paper will extend the Eisenberg-Noe network model approach of \cite{EN01}.  That paper considers the network of interbank obligations and finds the equilibrium payments.  
Central banks and regulators have applied the Eisenberg-Noe model to study cascading failures in the banking systems within their jurisdictions, see, e.g., \cite{ACP14,HK15,BEST04,ELS13,U11,G11,BBCH17}. 

The Eisenberg-Noe model has been extended previously to include more realistic structures for contagion; this includes bankruptcy costs, cross-holdings, and fire sales.  We refer to \cite{AW_15,Staum,H16} for surveys of these extensions.  
In our work we will consider specifically an extension related to fire sales.  Fire sales for a single (representative) illiquid asset have been studied in, e.g., \cite{CFS05,NYYA07,GK10,AFM13,CLY14,AW_15,AFM16}, and for multiple illiquid assets in, e.g., \cite{feinstein2015illiquid,feinstein2015leverage,feinstein2017multilayer}. 

The goal of this paper is to investigate the effects of confidence and liquidity on systemic risk and financial contagion.  To do so, a modified Eisenberg-Noe network model is considered under which there is a network of banks with connections between them representing interbank liabilities; additionally, the banks hold illiquid assets that may need to be liquidated in order to raise funds.  During the crisis events that are being studied, the asset selling is subject to price impacts from fire sales.  Moreover, firms are allowed to raise funds through short-term borrowing, in addition to liquidating assets.  The short-term interest rate is postulated to be a function of the confidence in the financial system and of the specific bank; as such the firms may have heterogeneous interest rates.  The focus of this paper is on the Nash equilibria of bank decisions over the two methods of raising cash.



The novelty of this project is in the consideration of other avenues of increasing reserve levels and raising cash. Specifically, one area that will be considered is borrowing. Most banks rely on short-term borrowing in order to fund their daily operations. The current project advances the field as it adds another dimension to systemic risk -- confidence, and specifically confidence as expressed through borrowing rates. 
While it is universally agreed that confidence is one of the most important pillars in finance, it is notoriously hard to quantify and predict. The problem consists of two main issues: How to quantify confidence in a measurable way; how to convert this quantity into cash flows, i.e.\ how to tie this quantity to a bank's cash flow. We use the price-to-book ratio (P/B) as a proxy for market confidence; see also \cite{sarin2016have} which was one of the first papers to include confidence when studying systemic risk.
We refer also to \cite{veraart2018distress} and \cite{barucca2016valuation} for a reduced form and structural model, respectively, that introduce notions of mark-to-market write-downs on interbank lending prior to actualized defaults; in doing so, these works incorporate the idea that confidence, or the lack thereof, can be an additional avenue for contagion.
As for the second issue the connection to the cash flow of a bank will be done by assuming that P/B is one of the variables that influences the short-term borrowing interest rate the bank has to pay to finance its daily operations. This has also been suggested in  \cite{GY14}, which cite liquidity and confidence as examples of additional channels that can generate substantial losses from contagion.  This is done in \cite{bichuch} which proposed a model to quantify and capture this effect in calculations of perceived riskiness of an individual bank. Endogenous interest rates in a financial network are computed in \cite{feinstein2018pricing}. This ties neatly to the other possibilities to raise cash to pay off debt. The ability to borrow, and liquidity in general, proved to be one of the major reasons of the financial crisis of 2007-2009, and one of the solutions employed by the regulators was to increase the liquidity of the overnight lending market.

This paper combines the above features in assessing the systemic risk in the financial system.  In particular, the proposed framework is general enough to incorporate all of the aforementioned concepts.  This model is an extension of the framework of \cite{EN01} incorporating borrowing, confidence, liquidity, and fire sales.  
The organization of this paper is as follows. 
Section \ref{sec:model} contains the details of the initial modification to the Eisenberg-Noe model that includes fire sales which we will be extending. Section \ref{sec:fire-sale} contains the analysis of optimal liquidation through fire sales and borrowing, together with some simple examples. Section \ref{sec:repo} incorporates the additional constraint that all borrowing must be collateralized.  Numerical case studies on both a symmetric system and one calibrated to European banking data are performed in Section \ref{sec:casestudies}. The proofs for all results are provided in the Appendix.

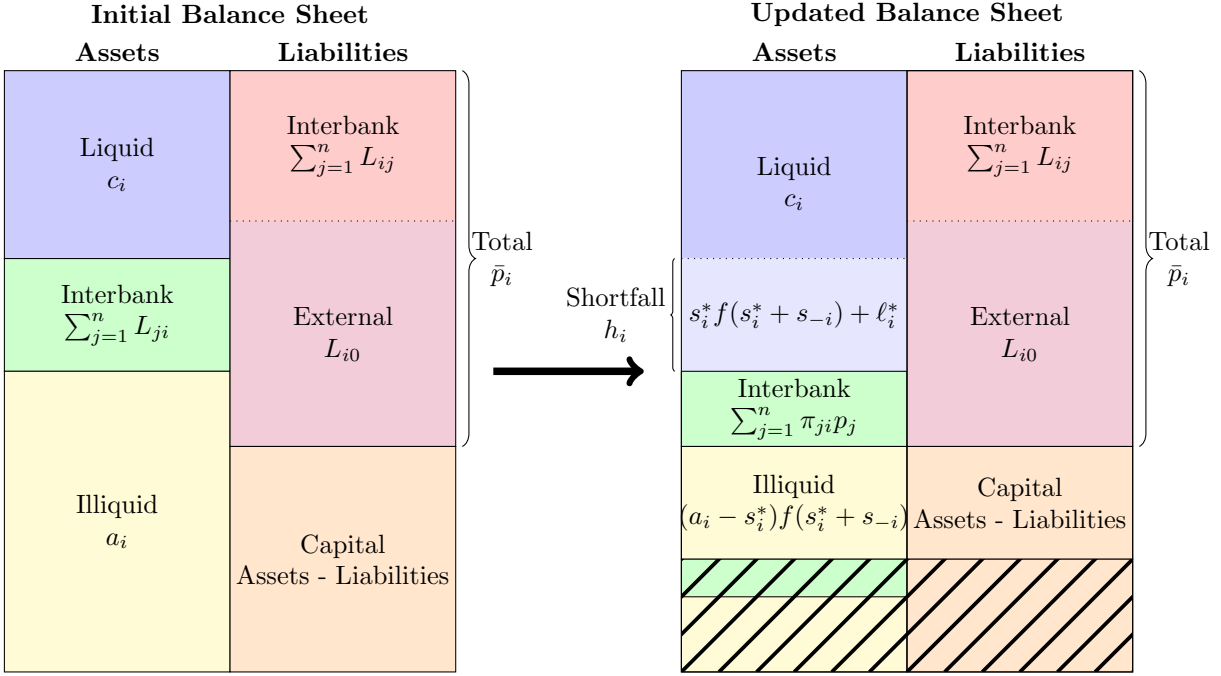
\begin{figure}
\centering
\begin{tikzpicture}
\draw[draw=none] (0,8.5) rectangle (6,9) node[pos=.5]{\bf Initial Balance Sheet};
\draw[draw=none] (0,8) rectangle (3,8.5) node[pos=.5]{\bf Assets};
\draw[draw=none] (3,8) rectangle (6,8.5) node[pos=.5]{\bf Liabilities};

\filldraw[fill=blue!20!white,draw=black] (0,5.5) rectangle (3,8) node[pos=.5,style={align=center}]{Liquid \\ $c_i$};
\filldraw[fill=green!20!white,draw=black] (0,4) rectangle (3,5.5) node[pos=.5,style={align=center}]{Interbank \\ $\sum_{j = 1}^n L_{ji}$};
\filldraw[fill=yellow!20!white,draw=black] (0,0) rectangle (3,4) node[pos=.5,style={align=center}]{Illiquid \\ $a_i$};

\filldraw[fill=red!20!white,draw=none] (3,6) rectangle (6,8) node[pos=.5,style={align=center}]{Interbank \\ $\sum_{j = 1}^n L_{ij}$};
\filldraw[fill=purple!20!white,draw=none] (3,3) rectangle (6,6) node[pos=.5,style={align=center}]{External \\ $L_{i0}$};
\draw[dotted] (3,6) -- (6,6);
\draw[decorate,decoration={brace,amplitude=5,mirror},xshift=2.5pt] (6,3) -- (6,8) node[pos=.5,right,style={align=center}]{Total \\ $\bar p_i$};
\draw (3,3) rectangle (6,8);
\filldraw[fill=orange!20!white,draw=black] (3,0) rectangle (6,3) node[pos=.5,style={align=center}]{Capital \\ Assets - Liabilities};

\draw[->,line width=1mm] (6.5,4) -- (8.5,4);

\draw[draw=none] (9,8.5) rectangle (15,9) node[pos=.5]{\bf Updated Balance Sheet};
\draw[draw=none] (9,8) rectangle (12,8.5) node[pos=.5]{\bf Assets};
\draw[draw=none] (12,8) rectangle (15,8.5) node[pos=.5]{\bf Liabilities};

\filldraw[fill=blue!20!white,draw=none] (9,5) rectangle (12,8) node[pos=.5,style={align=center}]{Liquid \\ $c_i$};
\filldraw[fill=blue!10!white,draw=none] (9,4) rectangle (12,5.5) node[pos=.5,style={align=center}]{$s_i^* f(s_i^* + s_{-i}) + \ell_i^*$};
\draw[dotted] (9,5.5) -- (12,5.5);
\draw[decorate,decoration={brace,amplitude=1.5},xshift=-2.5pt] (9,4) -- (9,5.5) node[pos=.5,left,style={align=center}]{Shortfall \\ $h_i$};
\filldraw[fill=green!20!white,draw=black] (9,3) rectangle (12,4) node[pos=.5,style={align=center}]{Interbank \\ $\sum_{j = 1}^n \pi_{ji} p_j$};
\filldraw[fill=yellow!20!white,draw=black] (9,1.5) rectangle (12,3) node[pos=.5,style={align=center}]{Illiquid \\ $(a_i-s_i^*) f(s_i^* + s_{-i})$};
\filldraw[fill=green!20!white,draw=black] (9,1) rectangle (12,1.5);
\filldraw[fill=yellow!20!white,draw=black] (9,0) rectangle (12,1);

\filldraw[fill=red!20!white,draw=none] (12,6) rectangle (15,8) node[pos=.5,style={align=center}]{Interbank \\ $\sum_{j = 1}^n L_{ij}$};
\filldraw[fill=purple!20!white,draw=none] (12,3) rectangle (15,6) node[pos=.5,style={align=center}]{External \\ $L_{i0}$};
\draw[dotted] (12,6) -- (15,6);
\draw[decorate,decoration={brace,amplitude=5,mirror},xshift=2.5pt] (15,3) -- (15,8) node[pos=.5,right,style={align=center}]{Total \\ $\bar p_i$};
\draw (12,3) rectangle (15,8);
\filldraw[fill=orange!20!white,draw=black] (12,1.5) rectangle (15,3) node[pos=.5,style={align=center}]{Capital \\ Assets - Liabilities};
\filldraw[fill=orange!20!white,draw=black] (12,0) rectangle (15,1.5);
\draw (9,0) rectangle (15,8);
\draw (12,0) -- (12,8);

\begin{scope}
    \clip (9,0) rectangle (15,1.5);
    \foreach \x in {-9,-8.5,...,15}
    {
        \draw[line width=.5mm] (9+\x,0) -- (15+\x,6);
    }
\end{scope}
\end{tikzpicture}
\caption{Stylized balance sheet for firm $i$ in Case III before and after payment and price updates. }
\label{fig:balance_sheet}
\end{figure}

\section{Original Eisenberg-Noe Fire Sales Model}\label{sec:model}
We wish to begin with some simple notation that will be utilized throughout this manuscript.  We will use the the notation that $[\cdot] _{i=1,..., n, j=1,...,m}$ and $[\cdot] _{i=1,..., n}$  to specify a $n\times m$ dimensional matrix and $n$-dimensional vector respectively, where $\diag\(   [\cdot] _{i=1,..., n}\)$ is a $n\times n$ matrix, with diagonal element $(i,i)$ is the $i$-th coordinate of the vector. Additionally, ${\bf{1}}_{n\times m}$ is a $n \times m$ matrix (and ${\bf{1}}_n$ is a $n \times 1$ vector) with all elements one.

Given $n$ interlinked banks, denote $L_{ij} \geq 0$ to be the liability of bank $i$ towards $j$, for $i,j=1,...,n$, and denote $\bar p_{i} = \sum_{j=0}^n L_{ij}$ to be the total liability of bank $i$. The liability $L_{i0},~i=1,..., n$ is assumed to be external liability of bank $i$ to an entity outside of the banking network.  It will be assumed that $L_{i0}\ge 0$. Under a pro-rata payment scheme, $[\Pi]_{ij}=\pi_{ij}$ will denote the relative liability, which is given by $\pi_{ij} = \frac{L_{ij}}{\bar p_i}$ if $\bar p_i > 0$ and $\pi_{ij} = 0$ otherwise, for $i=1,..., n, j=0, ..., n$. These define the $n\times (n+1)$ matrix $\Pi$. 
Moreover, it will be assumed that bank $i$ has liquid endowment $c_i \geq 0$ and illiquid endowment $a_i \geq 0$. While the liquid endowment can be assumed to be cash, the illiquid endowment is in physical units of assets as the liquidation price of these assets remains to be determined and will be assumed to depend on the size of the sale. For now, this price will be denoted by $q$. 

Following the network models of \cite{EN01,CFS05,AFM16,feinstein2015illiquid}, the notional payments are given by
$p = \bar p \wedge \left(c + s q + \Pi^\T p\right),$
where $s_i \in [0,a_i]$ is the quantity of illiquid assets being sold by firm $i$ evaluated by mark-to-market valuation with price $q$, and $c_i$ is its cash reserve. Throughout this work we will denote $x \wedge y = (\min(x_1,y_1),...,\min(x_n,y_n))^\T$ for $x,y \in \mathbb{R}^n$. There is an implicit no short selling constraint in this model.
It is assumed that the inverse demand curve $f$ for the illiquid asset provides the equilibrium price via
\begin{align}
q = f\(\sum_{i = 1}^n s_i\).
\label{eq:q}
\end{align}
The following are assumed about the inverse demand function $f$:
\begin{assumption}\label{ass:idf}
The inverse demand function $f: \R_+ \to [0,1]$ is strictly decreasing and twice continuously differentiable with $f(0) = 1$.
Let $M \geq \sum_{i = 1}^n a_i f(0) = \sum_{i = 1}^n a_i$ be the total initial market capitalization of the illiquid asset and $f(M) > 0$.
Additionally it will be assumed that the first derivative $f': \R_+ \to -\R_+$ is nondecreasing.  Further assume that the mapping $s \in [0,M] \mapsto s f(s)$ is strictly increasing and $s \in [0,M] \mapsto \frac{d^2}{ds^2}(s f(s)) = 2f'(s) + s f''(s)<0$ is strictly negative.
\end{assumption}

\begin{remark}
Here and elsewhere throughout this manuscript to simplify notation the one sided derivatives of $f$, such as the one calculated at the ends of the interval $[0,M]$, will be referred to as the derivative.
\end{remark}

The intuition being that during normal times, the price of the asset is one. However, during a fire sale the price of the illiquid asset is artificially depressed due to the lack of liquidity. Hence the book price of the asset is one, assuming the bank does not need to liquidate it, otherwise, the liquidation price will be set to be $q\in(0,1]$. 
Further, we incorporate the notion that the price drops slower at lower prices. Finally, if a bank were to sell an extra unit, it would obtain positive, but decreasing marginal, cash. This is consistent with the construction of an order book.


In this project the liquidation rule of \cite{AFM16} will be used.  That is, a firm must liquidate illiquid assets in order to have enough reserves to satisfy its liabilities, i.e., 
$s = a \wedge \frac{\bar p - c - \Pi^\T p}{q}.$
We note that, under a regular network \cite[Definition 5]{EN01} and Assumption~\ref{ass:idf}, the joint equilibrium between payments $p$ and prices $q$ is unique with this liquidation rule.

\section{Optimal Tradeoff between Fire Sales and Borrowing}\label{sec:fire-sale}
We now formulate an optimization problem from each bank's perspective. The bank may face a cash shortfall and has two avenues to raise funds: partially liquidate their holding in the illiquid asset, or borrow additional funds ($\ell_i$) in addition to $L_{i0}$ and temporarily delay the sale of the asset. The price-to-book (P/B) ratio is used as a proxy for the confidence in the bank. It should be noted that both the sale of the asset and borrowing of funds affect the price-to-book ratio. 

Assume that bank $i$'s interest rate is $r_i$, which may be a function of parameters such as the LIBOR rate. For example, these parameters include quantifiable proxies of confidence as suggested in \cite{bichuch}. However, since the model is static, these parameters will be assumed to be fixed, and a static interest rate will be assumed.
Additionally, it is convenient to define 
$s_{-i} = \sum_{j=1,j\ne i}^n s_j.$
Thus there are three cases to consider for each bank  $i$:
\begin{description}
\item[Case I:] {\it Bank $i$ is fundamentally insolvent.} In this case, the book value of the bank is below its obligations, that is, even if the bank were to sell its illiquid asset $a_i$ at the hypothetical price $f(0) = 1$ that prevails during normal, non-fire sale times it will still default on its debt. That is
\begin{align}
\bar p_i >  c_i + a_i + \sum_{j = 1}^n L_{ji}.
\label{eq:case2}
\end{align}
There is no borrowing in this case, as it will be assumed that it is imprudent to lend money to a fundamentally insolvent bank. Including a notion of bankruptcy costs as in, e.g., \cite{RV13,CK18}, it is assumed that for the duration of the crisis being studied, no obligations will be paid by any fundamentally insolvent banks, i.e.\ $p_i = 0$ for any bank in Case I. Such banks will not participate in fire sales as their assets will be distributed or liquidated in bankruptcy court at a later time.  As we allow all other banks to borrow to cover their liabilities, we will assume $p_i = \bar p_i$ for all banks that are \emph{not} in Case I.
\item[Case II:] {\it Bank $i$ is solvent with no borrowing nor asset liquidation is required.} This is the case if
\begin{align}
\bar p_i  \le  c_i + \sum_{j = 1}^n \pi_{ji} p_j.
\label{eq:case1}
\end{align}
\item[Case III:] {\it Bank $i$ is fundamentally solvent with borrowing and asset liquidations required.}  This is the case if
\begin{align}
c_i + \sum_{j = 1}^n \pi_{ji} p_j < \bar p_i \le  c_i + a_i + \sum_{j = 1}^n L_{ji}.
\label{eq:case3}
\end{align}
In this case, the bank can decide how much to borrow and liquidate in order to optimize its cash flow, i.e.\ minimize expenses due to the interest payment and the loss due to fire sale. In this case, it will be assumed that the bank can borrow funds as needed. The set of all such banks will be denoted $C_3$. 
\end{description}
Borrowing can, and will, only happen in the last case, when the bank is fundamentally solvent and at the asset's nominal price of $1$.  As such, the remainder of this section, and much of this paper will focus on banks in Case III.
Namely, bank $i$ in Case III seeks to optimally liquidate
\begin{align}
s_i^*(s_{-i}) = \argmin_{s\in[0,a_i]}  s (1-f(s_{-i} + s)) +  r_i \(h_i -s f(s_{-i} + s) \)^{+} 
\label{optim-tmp}
\end{align}
units of the illiquid asset
where $h_i = \bar p_i - \(c_i  + \sum_{j}\pi_{ji} p_j\)$ is the liquid shortfall for bank $i$.  Note that the total required borrowing $\ell_i$ given sales of $s_i$ is $\ell_i=(h_i - s_i f(s_{-i} + s_i))^+$. We are now ready to show the existence of Nash equilibrium. The proof of the theorem is given in Appendix \ref{app1}.

\begin{theorem}[Existence of Nash Equilibrium]\label{thm:exist}
Under Assumption \ref{ass:idf} there exists Nash equilibrium liquidating strategy $s^{**} = \(s_1^*(s_{-1}^{**}) , ... , s_n^*(s_{-n}^{**})\)^\T \in \prod_{i = 1}^n [0,a_i]$ with resulting equilibrium price $q^{**} = f\(\sum_{i = 1}^n s_i^{**}\).$
\end{theorem}

We now turn our attention to uniqueness of the Nash equilibrium strategy.  In fact, we can reformulate the optimal liquidation problem by the equivalent minimization problem
\begin{align}
s_i^*(s_{-i}) &= \argmin_{s \in [0,a_i]} \left\{s(1 - f(s_{-i} + s)) + r_i \(h_i - s f(s_{-i} + s)\) \; | \; s f(s_{-i} + s) \leq h_i\right\}, ~i \in C_3. \label{eq:Nash1}
\end{align}
Now let's consider a modified equilibrium problem between liquidations and the resultant prices, where the price in the constraint is replaced with the variable $q \in (0,1]$.
\begin{align}
s_i^{\dagger}(s^{\ddagger},q) &= \argmin_{s \in [0,\min\{a_i,\frac{h_i}{q}\}]} s\(1 - f\(\sum_{j \neq i} s_j^{\ddagger} + s\)\) + r_i \(h_i - s f\(\sum_{j \neq i} s_j^{\ddagger} + s\)\), ~i \in C_3.\label{eq:Nash3}
\end{align}
The following theorem  now asserts the uniqueness of the Nash Equilibrium. The proof is given in Appendix \ref{app2}.
\begin{theorem}[Uniqueness of Nash Equilibrium]\label{thm:unique}
Under Assumption~\ref{ass:idf} and for a fixed price $q \in [f(M),1]$ there exists a unique equilibrium liquidation strategy $\bar s(q) = s^{\dagger}(\bar s(q),q)$.  Additionally, if 
\begin{align}
f'(s) + s f''(s) \leq 0 \mbox{ and }-M f'(0) < f(M)
\label{eq:add-conditn}
\end{align} 
then 
there exists a unique joint liquidation-price equilibrium $s^{\ddagger} = s^{\dagger}(s^{\ddagger},q^{\ddagger})$ and $q^{\ddagger} = f(\sum_{j = 1}^n s_j^{\ddagger})$.
\end{theorem}

In the following examples, we show that the conditions \eqref{eq:add-conditn} of Theorem \ref{thm:unique} are true for a broad range of inverse demand functions $f$.  

\begin{example}\label{ex1}\textit{Linear price impact}:
$f(s) = 1 - \alpha s$ for $0 < \alpha < \frac{1}{2M}$ satisfies all conditions for an inverse demand function in Assumption \ref{ass:idf}.
The first condition of \eqref{eq:add-conditn} is true for $\alpha > 0$. While the second condition of \eqref{eq:add-conditn} is true when $\alpha < \frac{1}{2M}$.  That is, existence and uniqueness are guaranteed if $0 < \alpha < \frac{1}{2M}$.

\end{example}
\begin{example}\label{ex2}\textit{Exponential price impact}:
$f(s) = \exp(-\alpha s)$ for $0< \alpha < \frac{1}{M}$ satisfies all conditions for an inverse demand function in Assumption \ref{ass:idf}.  
The first condition of \eqref{eq:add-conditn} is true when $0 < \alpha \leq \frac1M.$ 
While the second condition of \eqref{eq:add-conditn} is true when $\alpha < \frac{W(1)}{M} \approx \frac{0.567}{M}$ where $W$ is the Lambert W function.  That is, existence and uniqueness are guaranteed if $0< \alpha < \frac{W(1)}{M}$.

\end{example}
\begin{example}\label{ex3}\textit{Hyperbolic price impact}:
$f(s) = \frac{\epsilon}{\epsilon + s}$ for $\epsilon > 0$ satisfies all conditions for an inverse demand function in Assumption \ref{ass:idf}.
The first condition of \eqref{eq:add-conditn} is true when $\epsilon \geq M.$ 
While the second condition of \eqref{eq:add-conditn} is true when $\epsilon > \frac{1 + \sqrt{5}}{2}M$.  That is, existence and uniqueness are guaranteed if $\epsilon > \frac{1 + \sqrt{5}}{2}M$.

\end{example}

We conclude this section by considering an algorithm to construct the clearing liquidations and prices under the conditions of Theorem~\ref{thm:unique}.  Briefly, the approach is to consider nested loops: a primary loop for the fixed point iterations in the price $q$ and a secondary loop to construct the equilibrium liquidation strategy $\bar s(q)$ using the approach taken in \cite{rosen65}. Additional details are given in Appendix \ref{app3}.
\begin{algorithm}\label{alg:unique}
Consider the setting of Theorem~\ref{thm:unique}.  The unique equilibrium liquidations $s^{\ddagger}$ and price $q^{\ddagger}$ can be found by the following algorithm.
\begin{enumerate}
\item Partition the firms into Cases I, II, and III.
\item Define the mapping $g$ and its Jacobian $G$ by
    \begin{align*}
    g(s) &= \diag\(\[\ind_{\{i \in C_3\}}\]_{i = 1,...,n}\)\(\hat g(s) - \(\diag(\ind_{\{s \leq 0\}})\hat g(s)^+ - \diag(\ind_{\{s \geq \min(a,h/q)\}})\hat g(s)^-\)\),\\ 
    \hat g(s) &= \diag(1+r)^{-1}\(1 - \diag(1+r)[f(\sum_{i = 1}^n s_i) + s f'(\sum_{i = 1}^n s_i)]\),\\
    G(s) &= -\diag\(\[\ind_{\{i \in C_3\}}\]_{i = 1,...,n}\)\((I+{\bf 1}_{n \times n})f'(\sum_{i = 1}^n s_i)+\diag(s) {\bf 1}_{n \times n} f''(\sum_{i = 1}^n s_i)\).
    \end{align*}
\item Initialize iteration $k = 0$, $s_i^k = 0$ for all firms $i = 1,...,n$, and $q^k = f(\sum_{i = 1}^n s_i^k) = 1$.
\item\label{alg:convergence} Repeat until convergence of $q^k$:
    \begin{enumerate}
    \item Iterate $k = k+1$.
    \item Initialize $s^k = s^{k-1}$ and $v = g(s^k)$. 
    \item Repeat until $v$ converges to $0$:
        \begin{enumerate}
        \item Set $t = -\frac{v^\T G(s^k) v}{\left\|G(s^k) v\right\|_2^2}$.
        \item Update $s_i^k = \min\(s_i^k + t v_i \; , \; a_i \; , \; h_i/q^k\)^+$ for all firms $i = 1,...,n$.
        \item Update $v = g(s^k)$.
        \end{enumerate}
    \item Set $q^k = f\(\sum_{i = 1}^n s_i^k\)$.
    \end{enumerate}
\item Define the equilibrium liquidations $s^{\ddagger} = s^k$ and price $q^{\ddagger} = q^k$.
\end{enumerate}
\end{algorithm}

\begin{figure}
\centering
\includegraphics[width=.7\textwidth]{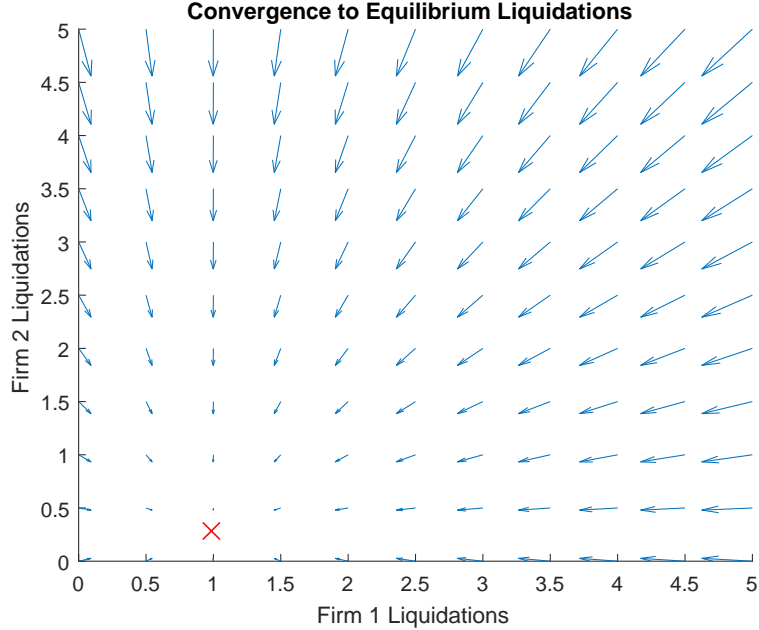}
\caption{A plot of speed and direction of the converging iterations of liquidations of two banks, based on the algorithm of \cite{rosen65}. The red `x' indicates the true equilibrium.}
\label{fig:alg}
\end{figure}
Figure \ref{fig:alg} graphically shows the convergence of the numerical iterations to find the Nash equilibrium in a two bank system with different initial points $s^0$.  For illustrative purposes we consider the network so that $h_1 = 4$, $h_2 = 2$, and $a_1 = a_2 = 5$, thus both firms are Case III institutions.  Further, we consider heterogeneous interest rates $r_1 = 12\%$ and $r_2 = 8\%$.  The market impacts are modeled by the linear inverse demand function $f(s) = 1-\alpha s$ with $\alpha = \frac{1}{21}$, which satisfies all assumptions of Theorem~\ref{thm:unique}.  Given the initial liquidations $(s_1^0,s_2^0) \in [0,5]^2$, the initial price $q^0 = f(s_1^0 + s_2^0)$ is found and the updated liquidations $\bar s(q^0)$ are computed.  We note that, all updated liquidations are in the direction of the clearing solution. The unique Nash equilibrium is shown as the red `x'.

\section{Optimal Fire Sales with Collateralized Borrowing}\label{sec:repo}
One of the key assumptions of the previous section was that banks can borrow funds without the need for collateral. In other words, once the determination is made that the bank is fundamentally solvent, there is no restriction on the size of the loan. The original interpretation can be that the loan is made by the lender of last resort, who initially determines if the bank is solvent or not. Once it is determined that the bank is solvent, the bank is allowed to borrow as needed. 

A more realistic and better alternative to consider is that the short-term loan needs to be collateralized, and that this collateral is in the form of the illiquid asset. The collateralized value of the illiquid asset will be assumed to be one. This is done in order to guarantee that a firm doing collateralized borrowing
will have positive equity using the  book value price for the illiquid asset.  More precisely, in Section \ref{sec:fire-sale}, it is possible that a firm that is required to both liquidate a portion of its holdings and borrow will fundamentally have negative equity after it liquidates assets and then borrows additional cash, which would be in contrast to the notion of ``solvency'' as it is generally considered.  Thus the current setup demonstrates more completely that firms required to liquidate and borrow remain fundamentally solvent, and bounds their actions so as to keep such a firm in this solvency region.
More specifically, this splits the original third case to be as follows: 
\begin{description}
\item[Case I:] {\it Bank $i$ is fundamentally insolvent.} Remains the same as \eqref{eq:case2}.  Similar to the original Case I, $p_i=0$ in this case, and otherwise banks will pay in full.
\item[Case II:] {\it Bank $i$ is solvent with no borrowing nor asset liquidation is required.} Also remains the same as \eqref{eq:case1}.
\item[Case III:] {\it Bank $i$ is fundamentally solvent, with borrowing and asset liquidations, but unable to pass a stress test.}  This is the case if
\begin{align}
c_i + \sum_{j = 1}^n \pi_{ji} p_j < \bar p_i \le  c_i + a_i + \sum_{j = 1}^n L_{ji},~\text{and}~a_i (1-\nu) < h_i.
\label{eq:case3-repo}
\end{align}
The first condition is the same as in the original case III, but an additional condition is now added, that under a stress test scenario, when the illiquid asset loses a proportion $\nu\in(0,1)$ of its value because of a shock, the bank becomes insolvent. It will be assumed that in this scenario this bank will be taken over by the regulator, who will also honor the bank's obligations and will ultimately be sold to a solvent bank. In this scenario the illiquid asset will not be liquidated.
\item[Case IV:]  {\it Bank $i$ is fundamentally solvent with borrowing and asset liquidations, and able to pass a stress test.} 
This is the case if
\begin{align}
c_i + \sum_{j = 1}^n \pi_{ji} p_j < \bar p_i \le  c_i + a_i + \sum_{j = 1}^n L_{ji},~\text{and}~a_i (1-\nu) \ge h_i.
\label{eq:case4-repo}
\end{align}
In this case, the bank can decide how much to borrow in order to optimize its cash flow, and minimize expenses due to the interest payment and the loss due to fire sale. The maximum borrowing amount is constrained by the collateral requirement. The set of all such banks will be denoted $C_4$.
\end{description}
As in the prior section, the last case of banks (Case IV firms) are those of primary interest for us. 
In this case, the value of the collateral will be assumed to be the asset's book price, i.e., one. Thus, for bank $i \in C_4$,  the maximum amount that can be borrowed is the solution to
\begin{align}
s^b_i(s_{-i}) (1 - f(s_{-i} +s^b_i(s_{-i}))) = a_i - h_i.
\label{eq:max-borrow}
\end{align}
The solution $s^b_i(s_{-i})$ to \eqref{eq:max-borrow} is unique if it exists, otherwise set $s_i^b(s_{-i}) = \infty$.
%
%
%

As such, the new problem that replaces \eqref{eq:Nash1} for bank $i$ in Case IV is to optimize the tradeoff in the number of sold shares and borrowing:
\begin{align}
s_i^*(s_{-i}) &= \argmin_{s \in [0,a_i]} \left\{s(1 - f(s_{-i} + s)) + r_i \(h_i - s f(s_{-i} + s)\) \; | \; s f(s_{-i} + s) \leq h_i,~s \le s^b_i(s_{-i})\right\}. \label{eq:optim}
\end{align}

The additional constraint in \eqref{eq:optim} as compared with \eqref{eq:Nash1} necessitates only trivial modifications to the proof of existence as done in Theorem \ref{thm:exist}. The following theorem is the analogue of Theorem \ref{thm:unique}, the proof of which is presented in Appendix \ref{app4}.

\begin{theorem}[Uniqueness of Nash Equilibrium]\label{thm:unique-repo}
Under Assumption~\ref{ass:idf} and for a fixed price $q \in [f(M),1]$ there exists a unique equilibrium liquidation strategy $\bar s(q) = s^{\dagger}(\bar s(q),q)$ where $s_i^{\dagger}(s^{\ddagger},q) = 0$ for $i \not\in C_4$ and 
\begin{align}
s_i^{\dagger}(s^{\ddagger},q) = \argmin\limits_{s \in [0,\min\{a_i,\frac{h_i}{q}, \frac{a_i - h_i}{1-q}\}]} s\(1 - f\(\sum_{j \neq i} s_j^{\ddagger} + s\)\) + r_i \(h_i - s f\(\sum_{j \neq i} s_j^{\ddagger} + s\)\) \quad \text{for } i \in C_4.
\label{eq:optim2}
\end{align}
Additionally, if 
$f'(s) + s f''(s) \leq 0 \mbox{ and } -M f'(0) < \nu \wedge f(M) $
then there exists a unique joint liquidation-price equilibrium $s^{\ddagger} = s^{\dagger}(s^{\ddagger},q^{\ddagger})$ and $q^{\ddagger} = f(\sum_{j = 1}^n s_j^{\ddagger})$.
\end{theorem}

\begin{remark}
Due to the constraints on liquidations and the possibility of borrowing we can immediately conclude that the (unique) equilibrium price provided by collateralized borrowing is greater than that in the uncollateralized case, which in turn is above the pure fire sale setting of \cite{AFM16}. 
Economically, this recovers the observations of \cite{GM09,B09} that the freezing of the repo market can generate excess systemic risk since we would, in some sense, move from the equilibrium under collateralized borrowing to the pure fire sale setting of \cite{AFM16}.
\end{remark}

\section{Case Studies}\label{sec:casestudies}
\subsection{Symmetric Case}

\subsubsection{Uncollateralized Borrowing:}
The simplest standard example is to consider the case when all banks are symmetric, i.e. their positions and shortfall are the same (see, e.g., \cite{ACM10,AOT13}). It is then possible to calculate the Nash equilibrium in Example \ref{ex1} explicitly.
There are again three cases to consider:
If each bank is in either Case I (fundamentally insolvent) or Case II (solvent with no borrowing or asset liquidations required), then the equilibrium is straight forward, and no computation is required. 

The interesting case is if every bank is in Case III (solvent with borrowing and/or asset liquidation is required). This implies that $h > 0$. In this case, there are two scenarios to consider: The scenario when all banks only liquidate and do not borrow, and the scenario when they both liquidate and borrow.
For convenience, since all the banks are symmetric, the index $i$ can be dropped.  Consider the linear price impact $f(s) = 1 - \alpha s$ from Example \ref{ex1} with $\alpha \in (0,\frac{1}{2M})$. Then, in the first scenario, from \eqref{eq:optim-sol} assuming that there are $n$ banks, we have that
$s^{L}  \(1 -  n \alpha s^{L} \)  =h.$
Hence,
\begin{align}
s^{L}_{\pm} = \frac{1 \pm \sqrt{1 - 4\alpha n h}}{2\alpha n}.
\label{eq:s-L}
\end{align}
If $s^{L}_{-}$ is real, then $s^{L}_{-} > 0$. In this case, the root $s^{L}_{-}$ is chosen as it is the minimal solution and all banks will liquidate the least amount of assets they need to. For convenience this solution will simply be denoted as $s^L$.  If $s^{L}_{-}$ is not real, then neither is $s^{L}_+$ and, as described in the proof of Theorem~\ref{thm:exist}, we set $s^L = +\infty$.

The second scenario is when the banks both liquidate and borrow. In this scenario, from \eqref{eq:s-0} it follows that
\begin{align}
1 - (1+r)\( 1 - \alpha (n+1) s^{0}\) = 0.
\end{align}
Hence,
\begin{align}
s^{0} = \frac{r }{\alpha(n+1)(1+r)}.
\label{eq:s0}
\end{align}
These scenario are determined by the size of the liquid shortfall $h$.  Specifically
$s^{L} < s^{0}$ if and only if $h  < \frac{1-\(1-2n \frac{r}{(1+r)(n+1)}\)^2}{4\alpha n}.$
Note that in this scenario $f(ns^{0}) = \(1-\alpha n s^{0} \) > \frac1{1+r}$, and hence the banks will never undertake a pure borrowing strategy.

\subsubsection{Collateralized Borrowing:}
We will assume $\nu > 0$ is set so that all banks are in Case IV as this is the interesting scenario.
In the setting that the borrowing needs to be  collateralized, the above calculations largely hold.  Namely, \eqref{eq:s-L} and \eqref{eq:s0} still hold. We need to add the constraint \eqref{eq:max-borrow} that $s^L,s^0\le s^b$, where
\begin{align}
s^b = \sqrt\frac{a-h}{\alpha n}.
\label{eq:sb}
\end{align} 
We then have:
\begin{itemize} 
\item $s^L\leq s^0\wedge s^b$ if and only if $h\leq \frac{1-\(1-2n \frac{r}{(1+r)(n+1)}\)^2}{4\alpha n} \wedge a(1-\alpha n a)$,
\item $s^0\leq s^L\wedge s^b$ if and only if $\frac{1-\(1-2n \frac{r}{(1+r)(n+1)}\)^2}{4\alpha n} \leq h \leq a-\frac{n r^2}{\alpha  (n+1)^2 (1+r)^2}$, and 
\item $s^b \leq s^L\wedge s^0$ if and only if $h \geq \frac{1-\(1-2n \frac{r}{(1+r)(n+1)}\)^2}{4\alpha n} \vee a(1-\alpha n a)$.
\end{itemize}

\subsubsection{Numerics:}

In the following two case studies, we consider an asset with market capitalization $M = 100$ and linear inverse demand function $f(s) = 1 - \frac{s}{210}$.  In the collateralized framework of Section \ref{sec:repo} we need to introduce the stress test parameter $\nu$, which we choose so that all banks are in Case IV.  Though we take parameter values that violate the sufficient conditions of Theorem~\ref{thm:unique-repo}, we find the clearing liquidations and price anyway as we have explicitly computed the symmetric Nash equilibrium above.

The left plot of Figure \ref{fig:graphs2} illustrates the effect of interest rate $r$ on the Nash equilibrium clearing price $q^\ddagger$. 
The case considered here is a symmetric system with $n = 90$ banks, each with shortfall $h = 1$, and with the market capitalization of the asset be $M = 100$. It is assumed that the assets are evenly distributed among all the banks, i.e., $a = M/n$.  In both the uncollateralized and collateralized settings, the price drops at the rate $\frac1{1+r} + \frac{r}{(n+1)(1+r)}$.  In the uncollateralized setting, the curve reaches the kink and flattens only when the firm would want to sell more assets than they have, i.e., $s^0 \geq a$.  In the collateralized setting, the clearing price reaches the kink and flattens at a lower interest rate (and higher clearing price) as the firms are constrained in their liquidations by the collateralization constraint, i.e., the firms wish to liquidate more than allowable for the repo market $s^0 \geq s^b$. This is consistent with intuition, as in a low interest rate environment the banks choose to borrow more, and liquidate less, than in a higher interest rate environment. The collateralization requirement has the effect that banks cannot dispose of too much of the illiquid asset, as its price would then drop too much otherwise, and they would not be able to use sufficiently collateralize their loans. Hence, the price flattens sooner than in the no collateralization case where the price continues to drop further. 

The right plot of Figure \ref{fig:graphs2} illustrates the effect of the liquid shortfall $h$ on the Nash equilibrium clearing price $q^\ddagger$. In this case, the system consists of $n = 50$ symmetric  banks, each with interest rate of $r = 20\%$, and with the market capitalization of the illiquid asset be $M = 100$. It is assumed that the assets are evenly distributed among all the banks, i.e., $a = M/n$.  We note that we constrain the liquid shortfall by the individual firm's assets $a$ so that the firm is fundamentally solvent.  Both the collateralized and uncollateralized settings are equivalent up until the firms are nearly at the boundary of fundamental solvency, i.e., $h \approx a$.  It is when $h \approx a$ that the collateralization constraint $s^\ddagger \leq s^b$ is relevant.  In particular, this demonstrates that only firms that are at risk of failing a stress test would be impacted by a repo market, and as such firms would be undeterred by the need to post collateral. This is consistent with the goal of the collateralization to secure the deal. As shortfall increases, banks become more and more risky. The banks would become fundamentally insolvent if their shortfall exceeded $h=2$. Hence, as the shortfall approaches $h=2$, the banks cannot sell as many assets as they would have done in the uncollateralized case, and need instead to borrow more.  This supports the price of the asset and limits the effects of the fire sale, while keeping the loans collateralized and secure.

\begin{figure}
\centering
\begin{subfigure}[b]{.45\textwidth}
\includegraphics[width=\textwidth]{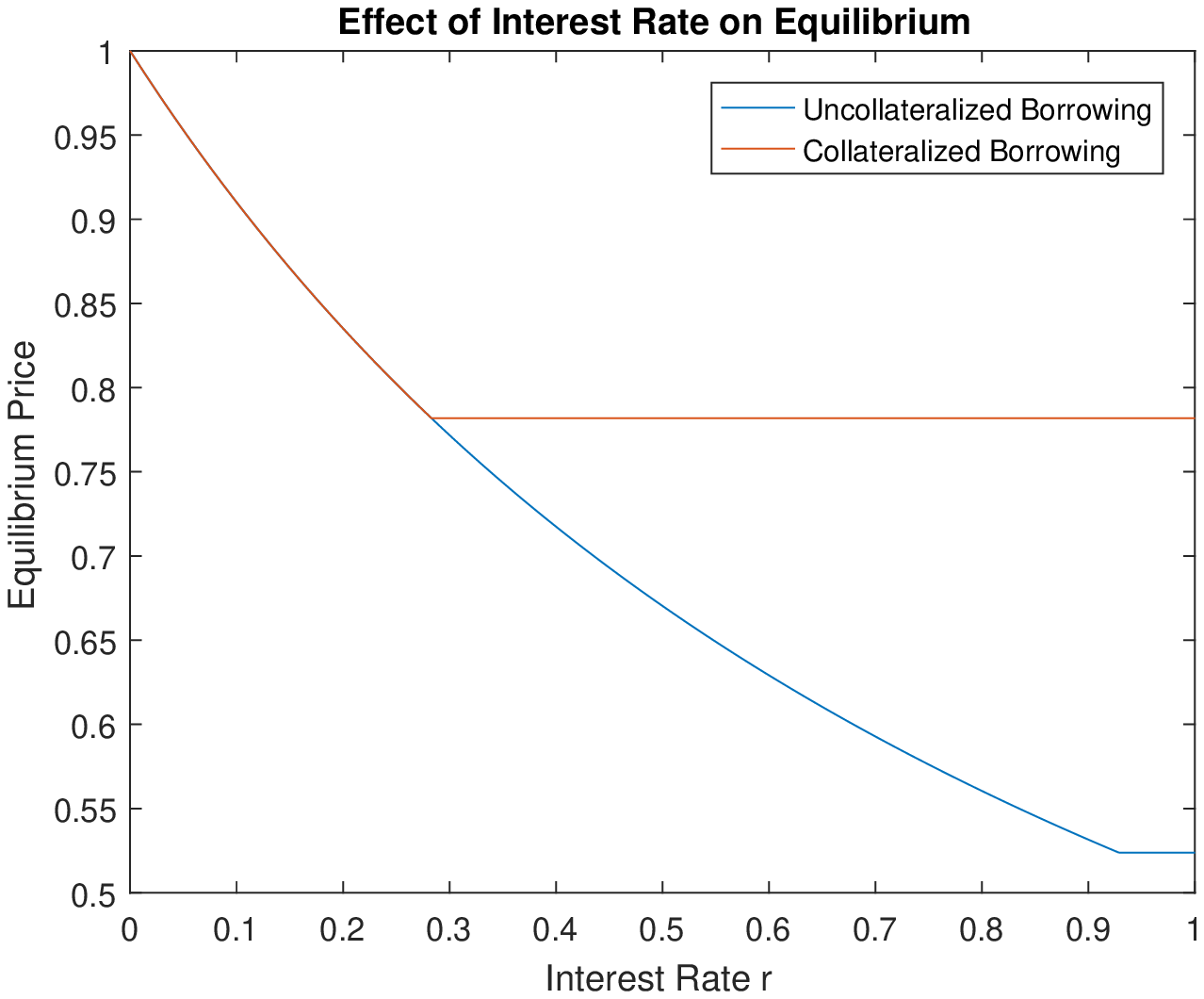}
\end{subfigure}
\begin{subfigure}[b]{.45\textwidth}
\includegraphics[width=\textwidth]{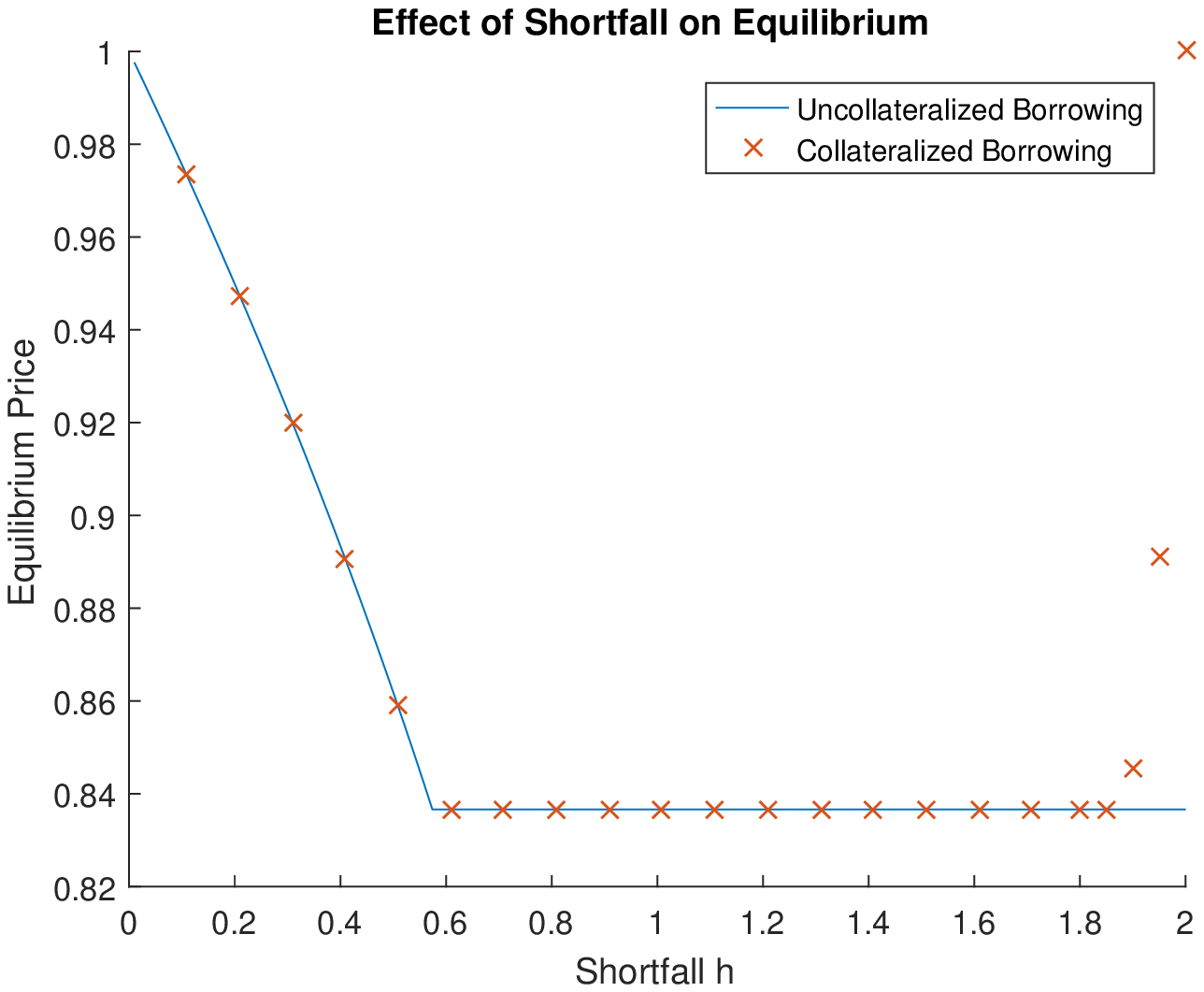}
\end{subfigure}
\caption{
Left: The effect of the interest rate $r$ on the clearing price in a symmetric financial system under both uncollateralized and collateralized borrowing.
Right:  The effect of the liquid shortfall $h$ on the clearing price in a symmetric financial system under both uncollateralized and collateralized borrowing.
}
\label{fig:graphs2}
\end{figure}

\subsection{EBA Case Study}
Let us now consider an example calibrated from 2011 European banking data that has been used in prior studies (e.g., \cite{GV16,CLY14}) under the financial contagion framework of \cite{EN01}.   We utilize the same network calibration as taken in \cite{feinstein2017multilayer}.  As in that work, we note that though we utilize this EBA dataset to have a more realistic network, the network calibration still requires heuristics and thus this example is still for demonstrative purposes only.\footnote{Due to complications with the calibration methodology, we only consider 87 of the 90 institutions. DE029, LU45, and SI058 were not included in this analysis.}

As a stylized bank balance sheet, we will consider three categories of assets: \emph{interbank assets} $\sum_{j = 1}^n L_{ji}$, \emph{liquid assets} $c_i$, and \emph{illiquid assets} $a_i$. We will additionally consider three categories of liabilities: \emph{interbank liabilities} $\sum_{j = 1}^n L_{ij}$, \emph{external liabilities} $L_{i0}$, and \emph{capital} $C_i$.  We refer back to Figure \ref{fig:balance_sheet} for demonstration of these terms. 

The dataset on European banks provides the total assets $T_i$, capital $C_i$, and interbank liabilities $\sum_{j = 1}^n L_{ij}$ for each bank $i$.  In order to use this dataset for our purposes we need to consider a few assumption.  First, as in \cite{CLY14,GY14}, we assume that the interbank liabilities equal interbank assets $\sum_{j = 1}^n L_{ij} = \sum_{j = 1}^n L_{ji}$ (with a small perturbation as required by the methodology of \cite{GV16}).  Additionally, as done in \cite{feinstein2017multilayer}, all assets not a part of the interbank assets are external assets (liquid and illiquid) and all liabilities not interbank nor capital are owed to the societal node $0$, i.e., external. Finally, the external assets are split into the component liquid and illiquid parts in proportion to the tier 1 capital ratio $R_i$. Under these assumptions, given the provided values, we determine the remainder of the stylized balance sheet via
\begin{align*}
c_i &= R_i\(T_i - \sum_{j = 1}^n L_{ij}\), \quad a_i = (1-R_i)\(T_i - \sum_{j = 1}^n L_{ij}\), \quad L_{i0} = T_i - \sum_{j = 1}^n L_{ij} - C_i, \quad \bar p_i = L_{i0} + \sum_{j = 1}^n L_{ij}.
\end{align*}
Under this calibration, the net worth of firm $i$ is equal to its capital, i.e., $C_i = T_i - \bar p_i$.
Finally, we construct the full nominal liabilities matrix $L$ using the methodology of \cite{GV16}.


In order to complete our model, we need to consider the remaining parameters of the system.  We set the market capitalization $M = \sum_{i = 1}^n a_i$ to be the total number of shares of the illiquid assets held by the banks.  Additionally we assume that all banks are subject to a $r = 5\%$ interest rate.  In this example we will focus solely on the linear inverse demand function $f(s) = 1 - \alpha s$ where we will vary $\alpha \in (0,\frac{1}{2M})$.  Finally, the stress test parameter is set to $\nu = 1\%$. As in the symmetric setting, though many of price impact levels $\alpha$ tested violate the conditions of Theorem \ref{thm:unique-repo} (i.e., $\alpha > \nu$), we find convergence of the numerical algorithms still hold and present the results herein.

We illustrate the two proposed enhancements to the simple liquidation strategy in a numerical study. The results are given in Figure \ref{fig:graphs}. The left figure illustrates both the equilibrium prices and the losses as a function of the price impacts. For this study these losses are defined as the losses of the illiquid asset due to fire sale, at a price below the original market price $f(0)=1$. The figure on the right on Figure \ref{fig:graphs} illustrates the number of insolvencies as a function of the price impacts. Both these figures were obtained in the linear inverse demand case, as given in Example \ref{ex1}. These figures are used to illustrate the amount of financial contagion in the three cases
-- the pure fire sale model of \cite{AFM16}, the uncollateralized borrowing case proposed in Section \ref{sec:fire-sale}, and the collateralized borrowing case considered in Section \ref{sec:repo}.

As expected, in all three cases, the equilibrium prices decline, as the price impact increases. This decline is most prominent in the basic case with no borrowing and only fire sales, and is very slight in either case when borrowing is allowed. This is also expected, as in both cases, when borrowing is allowed the banks will rely on borrowing more and more as the fire sale asset price declines. 
This is consistent with the graphs of the losses on the right axes. In both cases with borrowing, the losses become almost flat, as the price impact increases, also signifying that the banks rely more and more on borrowing, and are limiting the amount of capital lost due to the fire sale. This is of course not possible in a pure fire sale setting, in which case the losses increase steeply as the equilibrium price declines.
Additionally, we observe that the losses in the collateralized borrowing case are lower than the losses in the uncollateralized borrowing setting. This is contrary to intuition, as one 
expects that the losses would be larger in the collateralized borrowing setting, as the optimization performed by the banks in this case is more constrained than in the original uncollateralized setting. It turns out that in this case, because more constraints are introduced, the price remains higher than in the uncollateralized setting. As a consequence, the losses are smaller when collateralized borrowing is enforced.

The graph on the right of Figure \ref{fig:graphs} shows that once borrowing is introduced, neither uncollateralized nor collateralized borrowing experience any defaults or insolvencies.  In contrast, in the pure fire sale model of \cite{AFM16}, the number of defaults steadily rises as the price impact increases.  This occurs up to the point that all banks in the dataset are insolvent. This illustrates that our model of fire sales with borrowing more accurately encapsulates the health of the system than a pure fire sale model.  In particular, so long as borrowing markets remain liquid, we found that the European banking system would not be subjected to a systemic event in 2011 and all banks would survive (as evidenced in practice).

\begin{figure}
\centering
\begin{subfigure}[b]{.45\textwidth}
\includegraphics[width=\textwidth]{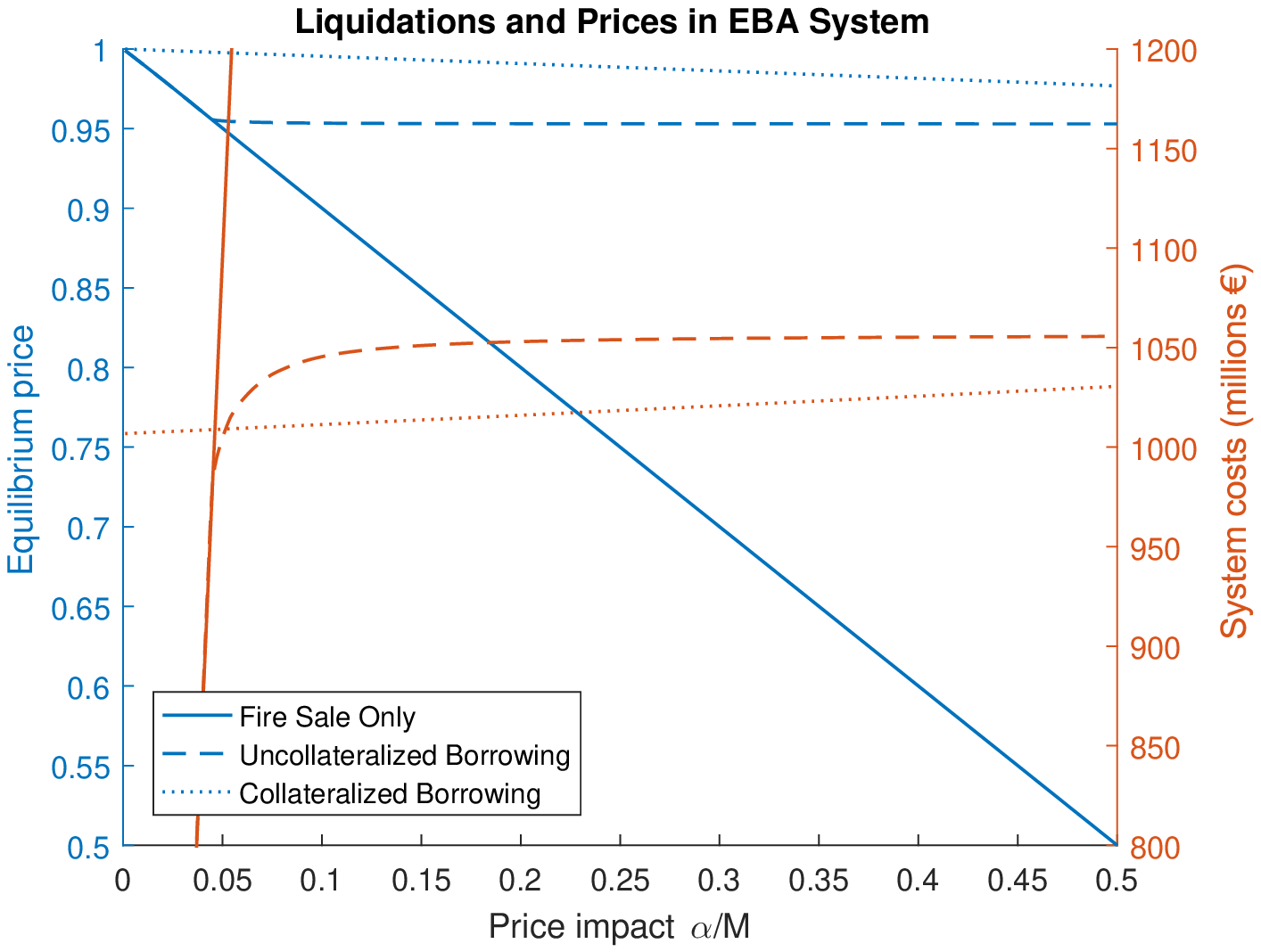}
\end{subfigure}
\begin{subfigure}[b]{.45\textwidth}
\includegraphics[width=\textwidth]{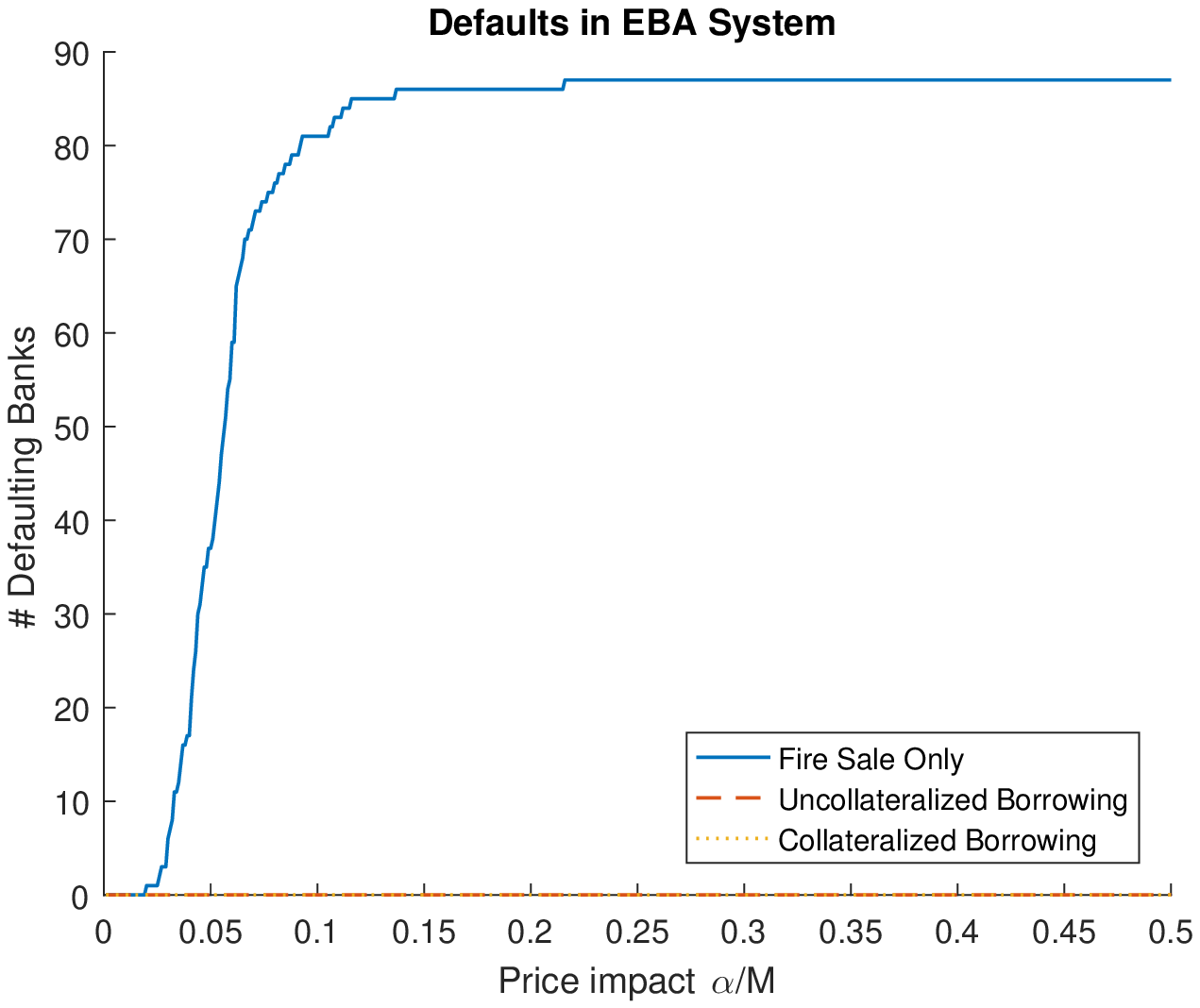}
\end{subfigure}
\caption{Left: Graph of equilibrium price (left axes) and losses (right axes) as a function of price impact. Right: Number of defaults as a function of price impact
}
\label{fig:graphs}
\end{figure}

\appendix
\section{Proof of Theorem~\ref{thm:exist}}\label{app1}
\begin{proof}[Proof of Theorem~\ref{thm:exist}]
Clearly, for $i\not\in C_3$, $s^{**}_{i}=0.$ For $i\in C_3$, two distinct scenarios will be considered:
The first scenario we consider is when the bank does not borrow and only liquidates. In this case 
$h_i = s f(s_{-i} + s).$  
The solution $s_i^{L}(s_{-i})$, if it exists, is unique and satisfies
\begin{align}
\label{eq:optim-sol}
s_i^{L}(s_{-i}) f(s_{-i} + s_i^{L}(s_{-i})) = h_i. 
\end{align}
If no solution to \eqref{eq:optim-sol} exists then $s_i^{L}(s_{-i}) = +\infty$.
Under the condition that the bank holds enough assets to possibly cover their entire shortfall, the existence and the uniqueness $s_i^{L}(s_{-i})$ follows from Assumption \ref{ass:idf}.

The second scenario considered here is when the bank does a mix of borrowing and liquidations.  From Assumption \ref{ass:idf}, this value can be found by equating the derivative of \eqref{optim-tmp} to zero and solving for $s_i^0(s_{-i})$ which satisfies
\begin{align}
1 - (1+r_i)(f(s_{-i} + s_i^0(s_{-i})) + s_i^0 f'(s_{-i} + s_i^0(s_{-i}))) = 0,
\label{eq:s-0}
\end{align}
where $s_i^0(s_{-i}) =+\infty$ if no such solution exists. Note that \textit{a priori} this might not be a solution to the optimization problem \eqref{optim-tmp}, as it additionally needs to be compared with $s^L_i(s_{-i})$ from \eqref{eq:optim-sol} as well as the bounds $0$ and $a_i$.
\begin{itemize}
\item If $f(s_{-i}) \geq (1 + r_i)^{-1}$ then the solution \eqref{eq:s-0}, if it exists, is unique since, by Assumption \ref{ass:idf}, $\operatorname{card}\{s \ge0 \; | \; f(s_{-i} + s) + s f'(s_{-i} + s) = (1 + r_i)^{-1}\} \le1$ in this case. Moreover, the solution to the original optimization problem \eqref{optim-tmp} is given by 
$s_i^*(s_{-i}) = \min(s_i^L(s_{-i}) , s_i^0(s_{-i}), a_i)$.
If $s_i^0(s_{-i}) \in [0,s_i^L(s_{-i})\wedge a_i]$ then it is optimal and the bank should liquidate $s_i^0(s_{-i})$ shares of the illiquid asset and borrow the remainder of the liquid shortfall.  If $s_i^0(s_{-i}) > s_i^L(s_{-i})\wedge a_i$ then it follows that 
\begin{align}
f(s_{-i} + s) + s f'(s_{-i} + s) > (1 + r_i)^{-1}~\mbox{ for every } s \in [0,s_i^L(s_{-i}) \wedge a_i].
\label{eq:inequl-s}
\end{align}
 Thus if $ s_i^L(s_{-i}) \le a_i$, the bank only needs to liquidate $s_i^L(s_{-i})$ number of illiquid assets, and no borrowing is needed. Whereas if $s_i^L(s_{-i}) > a_i$ the bank will liquidate all of its illiquid asset holdings $a_i$ and borrow the rest. This is the optimal behavior by \eqref{eq:inequl-s}.
From a financial perspective, it is optimal to liquidate as little as possible, so it is sufficient to liquidate $s_i^L(s_{-i})\wedge a_i$ if $s_i^0(s_{-i}) > s_i^L(s_{-i})\wedge a_i.$
%
%
%
\item If $f(s_{-i}) < (1 + r_i)^{-1}$ then $s_i^*(s_{-i}) = 0$, i.e. the optimal solution $s_i^{*} $  of \eqref{optim-tmp} is pure borrowing, and liquidating nothing.  This easily follows as for $s \in [0,a_i]$ the objective value is nondecreasing in $s$:
$$1 - (1 + r_i)(f(s + s_{-i}) + s f'(s + s_{-i})) > 1 - (1+r_i)(1+r_i)^{-1} = 0.$$
\end{itemize}
To summarize, each bank $i$ will choose to liquidate $s_i^*(s_{-i})$ shares of the illiquid asset provided that in aggregate all other firms are liquidating $s_{-i}$, where $s_i^*(s_{-i})$ is given by:
\begin{align}
s_i^*(s_{-i}) = \begin{cases} \min(s_i^L(s_{-i}) , s_i^0(s_{-i}), a_i) & \text{if } f(s_{-i}) \geq (1 + r_i)^{-1} \text{ and } i \in C_3\\ 0 &\text{otherwise}
\end{cases}.
\label{eq:s*}
\end{align}
In fact, we can see that $s_i^*: [0,M-a_i] \to [0,a_i]$ is continuous due to continuity of each of its components. Indeed, $s_i^*$ is continuous as a function of $s_{-i}$ in the (possibly empty) regions where  $f(s_{-i}) < (1 + r_i)^{-1}$ and $f(s_{-i}) > (1 + r_i)^{-1}$. It also can be seen that $s_i^*(s_{-i})\searrow0$ as $f(s_{-i}) \searrow (1 + r_i)^{-1}$, which establishes the continuity at the point where  $ f(s_{-i}) = (1 + r_i)^{-1}.$ 
Therefore, by Brouwer's fixed point theorem there exists an equilibrium liquidation strategy $s^{**} = \(s_1^*(s_{-1}^{**}) , ... , s_n^*(s_{-n}^{**})\)^\T \in \prod_{i = 1}^n [0,a_i]$.
\end{proof}

\section{Proof of Theorem~\ref{thm:unique}}\label{app2}
Before we can prove this theorem we require the following auxiliary lemma.
\begin{lemma}\label{lemma:diag-convex}
The function $H(s; \rho) = \sum_{i=1}^n H_i(s;\rho_i) ,~\rho \in\R_{+}^n$, is diagonally strictly convex where $H_i(s;\rho_i)=\rho_i \(s _i \(1-f\(\sum_{j=1}^n s_j \)\) + r_i \(h_i - s _i f\(\sum_{j=1}^n s_j \)\)\)$.
\end{lemma}
\begin{proof}
Recall from \cite{rosen65} that for $s\in\R^n$, the function $s\mapsto H(s;\rho)$ is diagonally strictly convex, if for some (fixed) $\rho\in\R_{+}^n$ and for every $s^0, s^1\in \R^n,~s^0 \neq s^1$, we have $(s^1 - s^0)^\T g(s^0; \rho) - (s^1 - s^0)^\T g(s^1; \rho) <0$, where 
\begin{align}g(s;\rho)   =  \(\begin{array}{c}  \d_{s_1} H_1(s;\rho_1)  \\  \vdots \\  \d_{s_n} H_n(s;\rho_n)\end{array}\) = \(\begin{array}{c} \rho_1 \d_{s_1} \(s _1 \(1-f\(\sum_{j=1}^n s_j \)\) + r_1 \(h_1 - s _1 f\(\sum_{j=1}^n s_j \)\)\) \\  \vdots \\ \rho_n \d_{s_n} \(s _n \(1-f\(\sum_{j=1}^n s_j \)\) + r_n \(h_n - s _n f\(\sum_{j=1}^n s_j \)\)\) \end{array}\).
\end{align}
Additionally, \cite[Theorem 6]{rosen65} shows that a sufficient condition for $H$ to be strictly convex is if $G(s;\rho) + G(s;\rho)^\T$ is a symmetric positive definite matrix for every $s \in \R^n$ and some $\rho \in \R_{+}^n$, where $G$ is the Jacobian matrix of $g$ with respect to $s$. 

Set $\rho_i = \frac1{1+r_i}$ then
\begin{align}
\[G(s;\rho) + (G(s;\rho))^\T\]_{ij} & = -\rho_i (1+r_i) \d_{s_is_j}^2\(s_i  f\(\sum_{k=1}^n s_k \)\) -\rho_j (1+r_j) \d_{s_js_i}^2\(s_j  f\(\sum_{k=1}^n s_k \)\) \\
&=- (2 + 2 \ind_{\{i=j\}}) f'\(\sum_{k=1}^n s_k \) - (s_i+s_j) f''\(\sum_{k=1}^n s_k \).
\end{align}
Thus, write $G(s;\rho) + G(s;\rho)^\T = G_1(s) + G_2(s) +f''\(\sum_{k=1}^n s_k \)  G_3(s),$
where 
\begin{align}
G_1(s) &= -\(2  f'\(\sum_{k=1}^n s_k \) +\(\sum_{k=1}^n s_k\) f''\(\sum_{k=1}^n s_k \) \){\bf{1}}_{n\times n},\\
G_2(s) & = -\diag\(2  f'\(\sum_{k=1}^n s_k \){\bf{1}}_{n} +s f''\(\sum_{k=1}^n s_k \) \),\\
[G_3(s)]_{ij} &= \sum_{k\ne i,j} s_k .
\end{align}

For any liquidations $s$, by construction, the matrix $G_1(s)$ is positive semidefinite and $G_2(s)$ is positive definite. 
Next, we show that $G_3(s)$ is positive semidefinite for any set of liquidations $s$, which would then give the desired result, as $f''\ge0$ by Assumption \ref{ass:idf}. In order to do that, note that $G_3(s) = \sum_{i=1}^n s_i {\bf{1}}_{n\times n}^{(i)},$ where the matrices ${\bf{1}}_{n\times n}^{(i)},~i=1, ..., n$ are given by $\[{\bf{1}}_{n\times n}^{(i)}\]_{jk} = \ind_{\{j\ne i, k\ne i\}}.$ It is readily seen that ${\bf{1}}_{n\times n}^{(i)}$ are positive semidefinite. It then follows that so is $G_3(s)$. 
\end{proof}

\begin{proof}[Proof of Theorem~\ref{thm:unique}]
We first fix $q \in [f(M),1],$ and look for an equilibrium $\bar s_i(q) = s_i^{\dagger}(\sum_{j\ne i}\bar s_j(q),q)$. That is, we look for the modified Nash equilibrium given by \eqref{eq:Nash3}.
For a fixed $q$, the existence of such an equilibrium follows from the logic of Theorem~\ref{thm:exist} and uniqueness of $\bar s(q)$ is a result of Lemma~\ref{lemma:diag-convex} as shown in \cite[Theorem 2]{rosen65}.

The next goal is to show $q \mapsto \Phi(q) =f(\sum_{j=1}^n \bar s_j(q))$ is a contraction mapping. 
Indeed, for $q_1 \neq q_2$:
\begin{align}
\frac{\abs{\Phi(q_1) -\Phi(q_2)}}{\abs{q_1 - q_2}} &= \frac{1}{\abs{q_1 - q_2}}\abs{f(\sum_{j=1}^n \bar s_j(q_1)) - f(\sum_{j=1}^n \bar s_j(q_2))} 
\le  -f'(0) \max_{q\in [f(M),1]} \abs{ \sum_{j=1}^n\bar s_j'(q)}.
\label{eq:contract1}
\end{align}
Thus to be a contraction mapping, it is sufficient to show that 
\begin{align}
-f'(0) \max_{q\in [f(M),1]} \abs{ \sum_{j=1}^n\bar s_j'(q)} <1.
\label{eq:contract2}
\end{align}

In order to show this, consider the sensitivity of $\bar s(q)$ with respect to $q$. Recall the construction of $s^{*}$ given by \eqref{eq:s*}; here we will replace $s_i^L(s_{-i})$ with $h_i/q$. Assume that $a_i, \frac{h_i}{q} , s_i^0(\sum_{j \neq i} \bar s_j(q) )$ are all different for all $i=1,..., n$, so that together with the continuity of $s^0$ it follows that $\bar s$ is differentiable with respect to $q$ and its derivative for a given bank $i$ is given by 
\begin{align}
\bar s_i'(q) = \ind_{\{i \in C_3\}}\(-\ind_{\{\frac{h_i}{q} < a_i \wedge s_i^0(\sum_{j \neq i} \bar s_j(q)) \}} \frac{h_i}{q^2} + (s_i^0)'(\sum_{j\ne i} \bar s_j(q)) (\sum_{j\ne i} \bar s_j'(q))  \ind_{\{0 \leq s_i^0 (\sum_{j \neq i} \bar s_j(q)) < \frac{h_i}{q} \wedge a_i  \} }    \).~~~~~~~~~
\label{eq:s-bar}
\end{align}
Here, 
the derivative of the optimal liquidation $s_i^0(s_{-i}), i=1, ...,n$ can be found via implicit differentiation by solving:
\begin{align}
&0 = (1 + (s_i^0)'(s_{-i})) f'(s_{-i} + s_i^0(s_{-i})) + (s_i^0)'(s_{-i}) f'(s_{-i} + s_i^0(s_{-i})) + s_i^0(s_{-i}) (1 + (s_i^0)'(s_{-i})) f''(s_{-i} + s_i^0(s_{-i}))\\
&= (s_i^0)'(s_{-i}) \(2 f'(s_{-i} + s_i^0(s_{-i})) + s_i^0(s_{-i}) f''(s_{-i} + s_i^0(s_{-i}))\) + \(f'(s_{-i} + s_i^0(s_{-i})) + s_i^0(s_{-i}) f''(s_{-i} + s_i^0(s_{-i}))\).
\end{align}
Thus by rearranging terms,
\begin{align}
(s_i^0)'(s_{-i}) &= -\frac{f'(s_{-i} + s_i^0(s_{-i})) + s_i^0(s_{-i}) f''(s_{-i} + s_i^0(s_{-i}))}{2f'(s_{-i} + s_i^0(s_{-i})) + s_i^0(s_{-i}) f''(s_{-i} + s_i^0(s_{-i}))}\label{eq:s0-prime}\\
&= \frac{f'(s_{-i} + s_i^0(s_{-i}))}{2f'(s_{-i} + s_i^0(s_{-i})) + s_i^0(s_{-i}) f''(s_{-i} + s_i^0(s_{-i}))} - 1.
\end{align}
Therefore  $(s_i^0)'(s_{-i})\in (-1,0]$ for all $i\in C_3,$ if $f'(x + s) \ge (2 f'(x + s) + s f''(x + s))$ for every $s \in [0,M-x]$ and any $x \in [0,M]$. It is sufficient to consider $x=0$, in other words, it is sufficient that the condition $f'(s) + s f''(s) <0$ holds.

Noting that $\bar s$ is zero for Case I and II institutions, solving the system \eqref{eq:s-bar}, it follows that 
\begin{align}
\bar s'(q) &= -\( I - \diag\( \[(s_i^0)'(\sum_{j\ne i} \bar s_j(q))  \ind_{\{0 \leq    s_i^0(\sum_{j \neq i} \bar s_j(q)) < \frac{h_i}{q} \wedge a_i  \} }\ind_{\{i \in C_3\}}   \]_{i = 1,...,n}\)\({\bf{1}}_{n\times n} - I\)\)^{-1} \\
&\qquad\qquad\times\diag\(\[\ind_{\{\frac{h_i}{q} < a_i \wedge s_i^0  (\sum_{j \neq i} \bar s_j(q)) \} }  \ind_{\{i \in  C_3 \}} \]_{i=1,...,n} \)\frac{h}{q^2}.~~~~~~~
\label{eq:bar-s'}
\end{align}
Using the fact that $(s_i^0)'(s_{-i})\in (-1,0]$ for $i=1,..., n$, it thus follows that 
\begin{align}
\left|{\bf{1}}_n^\T \bar s'(q)\right| \le \max_{d\in [0,1)^n} \left|{\bf{1}}_n^\T (I + \diag(d)({\bf{1}}_{n\times n} - I))^{-1} \diag\(\[\ind_{\{d_i=0,~\frac{h_i}{q} < a_i \}}\]_{i=1,...,n} \)\frac{h}{q^2}\right|.
\label{eq:bar-s'-bnd}
\end{align}
To compute this maximum, let $B(d) := I + \diag(d)({\bf{1}}_{n\times n} - I) = \diag\({\bf{1}}_{n} - d\) + d {\bf{1}}_{n}^\T$. 
By the Sherman-Morrison formula 
\begin{align}
&B(d)^{-1} = \diag\({\bf{1}}_{n}-d\)^{-1} -\frac1{1+{\bf{1}}_{n}^\T \diag\({\bf{1}}_{n}-d\)^{-1} d} \diag\({\bf{1}}_{n}-d\)^{-1}d {\bf{1}}_{n}^\T \diag\({\bf{1}}_{n}-d\)^{-1}\label{eq:B-invserse}\\
&= \(\begin{array}{llll}
\frac{1}{1-d_1} & 0 & \cdots & 0\\
0 & \frac{1}{1-d_2} & \cdots & 0 \\ 
\vdots & \vdots & \vdots & \vdots\\ 
0 & 0& \cdots & \frac{1}{1-d_n}
\end{array}\) -\frac1{1+ \sum_{k=1}^n \frac{d_k}{1-d_k}} 
\(\begin{array}{llll}
\frac{d_1}{(1-d_1)^2} & \frac{d_1}{(1-d_1)(1-d_2)} & \cdots & \frac{d_1}{(1-d_1)(1-d_n)}\\
\frac{d_2}{(1-d_2)(1-d_1)} & \frac{d_2}{(1-d_2)^2} & \cdots & \frac{d_2}{(1-d_2)(1-d_n)}\\
\vdots & \vdots & \vdots & \vdots\\ 
\frac{d_n}{(1-d_n)(1-d_1)} & \frac{d_n}{(1-d_n)(1-d_2)} & \cdots & \frac{d_n}{(1-d_n)^2}
\end{array}\).
\end{align}
It now follows that
\begin{align}
\label{eq:B-bnd}
&B(d)^{-1} \diag\(\[\ind_{\{d=0\}}\]_{i=1, .., n}\)\\
&= \frac1{1+ \sum_{k=1}^n \frac{d_k}{1-d_k}} \(\begin{array}{llll}
\ind_{\{d_1=0\}}(1+ \sum_{k=1}^n \frac{d_k}{1-d_k}) & -\frac{d_1\ind_{\{d_2=0\}}}{1-d_1} & \cdots & -\frac{d_1\ind_{\{d_n=0\}}}{1-d_1}\\
-\frac{d_2\ind_{\{d_1=0\}}}{1-d_2} & \ind_{\{d_2=0\}}(1+ \sum_{k=1}^n \frac{d_k}{1-d_k}) & \cdots & -\frac{d_2\ind_{\{d_n=0\}}}{1-d_2}\\
\vdots & \vdots & \vdots & \vdots\\ 
-\frac{d_n\ind_{\{d_1=0\}}}{1-d_n} &- \frac{d_n\ind_{\{d_2=0\}}}{1-d_n} & \cdots & \ind_{\{d_n=0\}}(1+ \sum_{k=1}^n \frac{d_k}{1-d_k})
\end{array}\)   .
\end{align}
And thus for any $j = 1,...,n$
\begin{align*}
\sum_{i = 1}^n \[B(d)^{-1}\]_{ij}\ind_{\{d_j = 0\}} = \frac{1}{1 + \sum_{k = 1}^n \frac{d_k}{1-d_k}}\(1 + \sum_{k = 1}^n \frac{d_k}{1-d_k} - \sum_{k \neq j} \frac{d_k}{1-d_k}\)\ind_{\{d_j = 0\}} = \frac{\ind_{\{d_j = 0\}}}{1 + \sum_{k = 1}^n \frac{d_k}{1-d_k}}.
\end{align*}
We conclude that 
\begin{align}
\max_{q\in [f(M),1]}\abs{{\bf{1}}_{n}^\T \bar s'(q)} & \le\max_{q\in [f(M),1],d\in [0,1)^n}\left\vert {\bf{1}}^\T_{n} B(d)^{-1} \diag\(\[\ind_{\{d_i=0,~\frac{h_i}{q} < a_i\} }\]_{i=1, ..., n} \)\frac{h}{q^2} \right\vert \label{eq:deriv-bnd}\\
&\le \max_{q\in [f(M),1]}\abs{{\bf{1}}_{n}^\T\frac{\frac{h}{q}\wedge a} {q}} \le \max_{q\in [f(M),1]} \frac{\sum_{i=1}^n a_i }q \le \frac{M}{f(M)}.
\end{align}
Recalling \eqref{eq:contract1}, we conclude that $\Phi$ is a contraction mapping if $-M f'(0) < f(M).$

Recall that it was assumed that $a_i, \frac{h_i}{q} , s_i^0(\sum_{j \neq i} \bar s_j(q) )$ are all different. If this assumption is violated, say $  s_i^0(\sum_{j \neq i} \bar s_j(q) )> a_i = \frac{h_i}{q}$, then we need to consider one-sided derivatives. In that case, the derivative from the left $\d_{-} \bar s_i(q) =0$, while the derivative from the right $\d_{+}\bar s_i(q) = -\frac{h}{q^2}.$ In this case, both one-sided derivatives would satisfy \eqref{eq:deriv-bnd}. The other cases can be treated similarly. 

\end{proof}

\section{Proof of Algorithm~\ref{alg:unique}}\label{app3}
\begin{proof}[Proof of Algorithm~\ref{alg:unique}]
The convergence of~\eqref{alg:convergence} follows from an outer convergence for the equilibrium price $q = f(\sum_{i = 1}^n \bar s_i(q))$ and an inner convergence to compute $s^k = \bar s(q^{k-1})$ for all iterations $k \geq 1$.  The outer convergence follows directly from the proof of Theorem~\ref{thm:unique} which shows that $f(\sum_{i = 1}^n \bar s_i(q))$ is a contraction mapping.  Theorem~10 of \cite{rosen65} provides the convergence of the inner loop to find $s^k$ as the equilibrium $\bar s(q^{k-1})$ for the fixed price level $q^{k-1}$.  In particular, this inner loop is a gradient descent algorithm for the coupled optimization problems~\eqref{eq:Nash3} through consideration of the KKT conditions.  The gradient and step-size, as constructed for Theorem~10 of \cite{rosen65}, are provided by the mapping $g$ and the value $t$ respectively.
\end{proof}

\section{Proof of Theorem~\ref{thm:unique-repo}}\label{app4}
\begin{proof}[Proof of Theorem~\ref{thm:unique-repo}]
This proof remains largely unchanged from those of Theorems~\ref{thm:exist} and \ref{thm:unique}. The mapping $q \mapsto \Phi(q) =f(\sum_{j=1}^n \bar s_j(q))$ is a contraction mapping if \eqref{eq:contract1} is dominated by $1$. We need to change $\bar s'(q)$ in \eqref{eq:s-bar} in the following manner. Similar to the original proof of Theorem \ref{thm:unique} we assume that  the quantities $a_i ,\frac{a_i-h_i}{1-q},  \frac{h_i}{q}, s_i^0(s_{-i})~i=1, ...,n $ are all different. The modification to the proof if that is not the case, is the same as in the original proof. Under this assumption 
\begin{align}
\bar s_i'(q) &= \ind_{\{i \in C_4\}} \Big(-\ind_{\{\frac{h_i}{q}\wedge\frac{a_i-h_i}{1-q} < a_i \wedge s_i^0(\sum_{j \neq i} \bar s_j(q))\}} \( \frac{h_i}{q^2}\ind_{\{\frac{h_i}{q} <\frac{a_i-h_i}{1-q}\}} -  \frac{a_i-h_i}{(1-q)^2}\ind_{\{\frac{h_i}{q} \ge\frac{a_i-h_i}{1-q}\}} \)\\
&+ (s_i^0)'(\sum_{j\ne i} \bar s_j(q)) (\sum_{j\ne i} \bar s_j'(q))  \ind_{\{0 \leq s_i^0(\sum_{j \neq i} \bar s_j(q)) < \frac{h_i\wedge(a_i-h_i)}{q} \wedge a_i\}}\Big), \mbox{ for }i=1,..., n. 
\end{align}
Then, similar to \eqref{eq:bar-s'}
\begin{align}
&\bar s'(q) \\
&= -\( I - \diag\( \[(s_i^0)'(\sum_{j\ne i} \bar s_j(q))  \ind_{\{0 \leq s_i^0(\sum_{j \neq i} \bar s_j(q)) < \frac{h_i}{q} \wedge \frac{a_i-h_i}{1-q}\wedge a_i\}}\ind_{\{i \in C_4\}}\]_{i=1, ..., n}\)\({\bf{1}}_{n\times n} - I\)\)^{-1} \\
&\times\diag\(\[\ind_{\{\frac{h_i}{q}\wedge\frac{a_i-h_i}{1-q} < a_i \wedge s_i^0(\sum_{j \neq i} \bar s_j(q))\}}\ind_{\{i \in C_4\}}\]_{i=1, ..., n} \) \[ \frac{h_i}{q^2}\ind_{\{\frac{h_i}{q} <\frac{a_i-h_i}{1-q}\}} -  \frac{a_i-h_i}{(1-q)^2}\ind_{\{\frac{h_i}{q} \ge\frac{a_i-h_i}{1-q}\}} \]_{i=1, ..., n}.
\label{eq:bar-s'-repo}
\end{align}
And similar to \eqref{eq:bar-s'-bnd} we recover 
\begin{align}
&\left|{\bf{1}}_n^\T \bar s'(q)\right| \\
&\le \max_{d\in [0,1)^n} \left|{\bf{1}}_n^\T B(d)^{-1} \diag\(\[\ind_{\{d_i = 0,~\frac{h_i}{q}\wedge\frac{a_i-h_i}{1-q} < a_i\}}\]_{i=1, ..., n} \) \[ \frac{h_i}{q^2}\ind_{\{\frac{h_i}{q} <\frac{a_i-h_i}{1-q}\}} -  \frac{a_i-h_i}{(1-q)^2}\ind_{\{\frac{h_i}{q} \ge\frac{a_i-h_i}{1-q}\}} \]_{i=1, ..., n}\right|
\label{eq:bar-s'-bnd-repo}
\end{align}
again letting $B(d) = I + \diag(d)({\bf{1}}_{n\times n} - I) $. We obtain the same bound \eqref{eq:B-bnd} on $B(d)^{-1} \diag\(\ind_{\{d=0\}}\)$.
We finally conclude 
\begin{align}
\max_{q\in [f(M),1]}&\abs{{\bf{1}}_{n}^\T \bar s'(q)}  \\
&\le\max_{\substack{q\in [f(M),1]\\d\in [0,1)^n}}\left\vert {\bf{1}}_{n}^\T B(d)^{-1} \diag\(\ind_{\{d = 0,~\frac{h}{q}\wedge\frac{a-h}{1-q} < a\}}\) \[ \frac{h_i}{q^2}\ind_{\{\frac{h_i}{q} <\frac{a_i-h_i}{1-q}\}} -  \frac{a_i-h_i}{(1-q)^2}\ind_{\{\frac{h_i}{q} \ge\frac{a_i-h_i}{1-q}\}} \]_{i=1, ..., n}\right\vert\\
&\le \max_{q\in [f(M),\max\limits_{i = 1,...,n}\frac{h_i}{a_i}]}\abs{{\bf{1}}_{n}^\T\frac{\frac{a-h}{1-q}\wedge a}{1-q} } \vee 
\max_{q\in [f(M),1]}\abs{{\bf{1}}_{n}^\T\frac{\frac{h}{q}\wedge a}{q}  } \\
&\le \frac{\sum_{i = 1}^n a_i}{\nu} \vee \frac{\sum_{i = 1}^n a_i}{f(M)}
\le \frac{M}{\nu\wedge f(M)}.
\label{eq:deriv-bnd-repo}
\end{align}

\end{proof}

\bibliographystyle{myplain}
\small{\bibliography{bibtex2}}

\begin{thebibliography}{10}

\bibitem{AOT13}
D.~Acemoglu, A.~Ozdaglar, and A.~Tahbaz-Salehi.
\newblock Systemic risk and stability in financial networks.
\newblock {\em American Economic Review}, 105(2):564--608, 2015.

\bibitem{ACM10}
H.~Amini, R.~Cont, and A.~Minca.
\newblock Resilience to contagion in financial networks.
\newblock {\em Mathematical Finance}, 26(2):329--365, 2016.

\bibitem{AFM13}
H.~Amini, D.~Filipovi\'{c}, and A.~Minca.
\newblock Systemic risk with central counterparty clearing.
\newblock Swiss {Finance} {Institute} {Research} {Paper} {No.} 13-34, Swiss
  Finance Institute, 2013.

\bibitem{AFM16}
H.~Amini, D.~Filipovi\'{c}, and A.~Minca.
\newblock Uniqueness of equilibrium in a payment system with liquidation costs.
\newblock {\em Operations Research Letters}, 44(1):1--5, 2016.

\bibitem{ACP14}
K.~Anand, B.~Craig, and G.~Von~Peter.
\newblock Filling in the blanks: Network structure and interbank contagion.
\newblock {\em Quantitative Finance}, 15(4):625--636, 2015.

\bibitem{feinstein2018pricing}
T.~Banerjee and Z.~Feinstein.
\newblock Pricing of debt and equity in a financial network with comonotonic
  endowments.
\newblock 2018.
\newblock Working paper.

\bibitem{BBCH17}
M.~Bardoscia, P.~Barucca, A.~Brinley~Codd, and J.~Hill.
\newblock The decline of solvency contagion risk.
\newblock {\em Bank of England Staff Working Paper}, 662, 2017.

\bibitem{barucca2016valuation}
P.~Barucca, M.~Bardoscia, F.~Caccioli, M.~D'Errico, G.~Visentin, S.~Battiston,
  and G.~Caldarelli.
\newblock Network valuation in financial systems.
\newblock 2016.
\newblock Working paper.

\bibitem{bichuch}
M.~Bichuch and K.~Chen.
\newblock {\em Systemic Risk: the Effect of Market Confidence}.
\newblock 2017.
\newblock Working paper.

\bibitem{BEST04}
M.~Boss, H.~Elsinger, M.~Summer, and S.~Thurner.
\newblock Network topology of the interbank market.
\newblock {\em Quantitative Finance}, 4(6):677--684, 2004.

\bibitem{B09}
M.~K. Brunnermeier.
\newblock Deciphering the liquidity and credit crunch 2007-2008.
\newblock {\em The Journal of Economic Perspectives}, 23(1):77--100, 2009.

\bibitem{CLY14}
N.~Chen, X.~Liu, and D.~D. Yao.
\newblock An optimization view of financial systemic risk modeling: The network
  effect and the market liquidity effect.
\newblock {\em Operations Research}, 64(5), 2016.

\bibitem{CK18}
C.~Chong and C.~Kl\"{u}ppelberg.
\newblock Contagion in financial systems: A {B}ayesian network approach.
\newblock {\em SIAM Journal on Financial Mathematics}, 9(1):28--53, 2018.

\bibitem{CFS05}
R.~Cifuentes, H.~S. Shin, and G.~Ferrucci.
\newblock Liquidity risk and contagion.
\newblock {\em Journal of the European Economic Association}, 3(2-3):556--566,
  2005.

\bibitem{EN01}
L.~Eisenberg and T.~H. Noe.
\newblock Systemic risk in financial systems.
\newblock {\em Management Science}, 47(2):236--249, 2001.

\bibitem{ELS13}
H.~Elsinger, A.~Lehar, and M.~Summer.
\newblock Network models and systemic risk assessment.
\newblock In {\em Handbook on Systemic Risk}, pages 287--305. Cambridge
  University Press, 2013.

\bibitem{feinstein2015illiquid}
Z.~Feinstein.
\newblock Financial contagion and asset liquidation strategies.
\newblock {\em Operations Research Letters}, 45(2):109--114, 2017.

\bibitem{feinstein2017multilayer}
Z.~Feinstein.
\newblock Obligations with physical delivery in a multi-layered financial
  network.
\newblock 2017.
\newblock Working paper.

\bibitem{feinstein2015leverage}
Z.~Feinstein and F.~El-Masri.
\newblock The effects of leverage requirements and fire sales on financial
  contagion via asset liquidation strategies in financial networks.
\newblock {\em Statistics \& Risk Modeling}, 34(3-4):109--114, 2017.

\bibitem{G11}
P.~Gai, A.~Haldane, and S.~Kapadia.
\newblock Complexity, concentration and contagion.
\newblock {\em Journal of Monetary Economics}, 58(5):453--470, 2011.

\bibitem{GK10}
P.~Gai and S.~Kapadia.
\newblock Contagion in financial networks.
\newblock Bank of England Working Papers 383, Bank of England, 2010.

\bibitem{GV16}
A.~Gandy and L.~A. Veraart.
\newblock A {B}ayesian methodology for systemic risk assessment in financial
  networks.
\newblock {\em Management Science}, 2016.
\newblock {DOI: 10.1287/mnsc.2016.2546}.

\bibitem{GY14}
P.~Glasserman and H.~P. Young.
\newblock How likely is contagion in financial networks?
\newblock {\em Journal of Banking and Finance}, 50:383--399, 2015.

\bibitem{GM09}
G.~B. Gorton and A.~Metrick.
\newblock {\em Journal of Financial Economics}, 104(3):425--451, 2012.

\bibitem{HK15}
G.~Halaj and C.~Kok.
\newblock Modelling the emergence of the interbank networks.
\newblock {\em Quantitative Finance}, 15(4):653--671, 2015.

\bibitem{H16}
A.-C. H\"{u}ser.
\newblock Too interconnected to fail: A survey of the interbank networks
  literature.
\newblock {\em Journal of Network Theory in Finance}, 1(3):1--50, 2015.

\bibitem{NYYA07}
E.~Nier, J.~Yang, T.~Yorulmazer, and A.~Alentorn.
\newblock Network models and financial stability.
\newblock {\em Journal of Economic Dynamics and Control}, 31(6):2033--2060,
  2007.

\bibitem{RV13}
L.~C. Rogers and L.~A. Veraart.
\newblock Failure and rescue in an interbank network.
\newblock {\em Management Science}, 59(4):882--898, 2013.

\bibitem{rosen65}
J.~B. Rosen.
\newblock Existence and uniqueness of equilibrium points for concave n-person
  games.
\newblock {\em Econometrica}, 33(3):520--534, 1965.

\bibitem{sarin2016have}
N.~Sarin and L.~H. Summers.
\newblock {\em Have big banks gotten safer?}
\newblock Brookings Institution, 2016.

\bibitem{Staum}
J.~Staum.
\newblock Counterparty contagion in context: Contributions to systemic risk.
\newblock In {\em Handbook on Systemic Risk}, pages 512--548. Cambridge
  University Press, 2013.

\bibitem{U11}
C.~Upper.
\newblock Simulation methods to assess the danger of contagion in interbank
  markets.
\newblock {\em Journal of Financial Stability}, 7(3):111--125, 2011.

\bibitem{veraart2018distress}
L.~A. Veraart.
\newblock Distress and default contagion in financial networks.
\newblock 2018.
\newblock Working paper.

\bibitem{AW_15}
S.~Weber and K.~Weske.
\newblock The joint impact of bankruptcy costs, fire sales and cross-holdings
  on systemic risk in financial networks.
\newblock {\em Probability, Uncertainty and Quantitative Risk}, 2(1):9, June
  2017.

\end{thebibliography}

\end{document}